\documentclass{article}

\usepackage{microtype}
\usepackage{graphicx}
\usepackage{subcaption}
\usepackage{booktabs} 

\usepackage{hyperref}

\usepackage[preprint]{icml2026}

\usepackage{amsmath}
\usepackage{amssymb}
\usepackage{mathtools}
\usepackage{amsthm}
\usepackage{bm}
\usepackage{xcolor}

\usepackage[capitalize,noabbrev]{cleveref}

\theoremstyle{plain}
\newtheorem{theorem}{Theorem}[section]
\newtheorem{proposition}[theorem]{Proposition}

\theoremstyle{definition}

\theoremstyle{remark}
\newtheorem{remark}{Remark}[section]

\newcommand{\argmin}[1]{\underset{#1}{\operatorname{argmin}}\;}

\definecolor{darkgreen}{rgb}{0.0, 0.5, 0.0}

\usepackage[textsize=tiny]{todonotes}

\icmltitlerunning{Joint Nuclear and $\ell_1$ Regularization for Logistic Matrix Regression}

\begin{document}

\twocolumn[
  \icmltitle{Joint Nuclear and $\ell_1$ Regularization for Logistic Matrix Regression with Applications to Brain Imaging}
%
%


  \icmlsetsymbol{equal}{*}

  \begin{icmlauthorlist}
    \icmlauthor{Damian Brzyski}{jag}
    \icmlauthor{Aaron Cohen}{bloom}
    \icmlauthor{Zijian Wang}{bloom}
    \icmlauthor{Mario Dzemidzic}{alc,rad}
    \icmlauthor{David A. Kareken}{alc}
    \icmlauthor{Jaroslaw Harezlak}{bloom}
  \end{icmlauthorlist}

  \icmlaffiliation{jag}{Center for Quantitative Research in Political Science, Jagiellonian University, Poland}
  \icmlaffiliation{bloom}{Department of Epidemiology and Biostatistics, Indiana University Bloomington, Bloomington, IN, USA}
  \icmlaffiliation{rad}{Department of Radiology and Imaging Sciences, Indiana University School of Medicine, Indianapolis, IN, USA}
  \icmlaffiliation{alc}{Department of Neurology, Indiana University School of Medicine, Indiana Alcohol Research Center, Indianapolis, IN, USA}
  
  \icmlcorrespondingauthor{Damian Brzyski}{damian.brzyski@uj.edu.pl}
  \icmlcorrespondingauthor{Jaroslaw Harezlak}{harezlak@iu.edu}

  \icmlkeywords{Logistic Regression, Scalar-On-Matrix Regression, Sparse and Low-Rank Regularization, ADMM, Brain Imaging, L1 Penalty, Nuclear Norm, Convex Optimization}

  \vskip 0.3in
]



\printAffiliationsAndNotice{}  

\begin{abstract}
    We introduce a new convex optimization framework for logistic scalar-on-matrix regression which incorporates nuclear and $\ell_1$ norm penalties to enforce simultaneously low-rank and sparse structures in the estimated coefficient matrix. The proposed method enables interpretable modeling of high-dimensional matrix-valued predictors in the presence of binary responses. We derive a custom algorithm based on the Alternating Direction Method of Multipliers (ADMM) to efficiently solve the resulting convex optimization problem and establish the theoretical properties of the obtained solution. Numerical experiments clearly demonstrate the effectiveness of our method in recovering meaningful predictive patterns.
    Finally, we apply our method to the brain imaging data to identify structures in functional brain connectivity matrices that are characteristic of subjects with a family history of alcohol use disorders (AUDs).
\end{abstract}

\section{Introduction}
In high-dimensional learning problems, particularly when modeling a scalar outcome from a matrix-valued input, placing suitable structural constraints is crucial for reliable estimation. Two widely used structures are \textit{sparsity}, which assumes only a small subset of features is relevant, and \textit{low-rankness}, which assumes model parameters lie near a low-dimensional subspace. Many real-world problems exhibit both properties simultaneously. For example, in face recognition, an image can be decomposed into a low-rank component that captures illumination and overall structure plus a sparse component for occlusions or noise \cite{wright2009robust}. In multi-task learning, a coefficient matrix for $K$ related tasks may have a shared low-dimensional basis while utilizing only a subset of features per task \cite{chen2011integrating}. In brain imaging, neurological conditions might affect a limited subset of brain connections that reflect a few underlying network-level patterns \cite{brzyski2024spinner}. These examples illustrate the appeal of combining sparsity and low-rankness to obtain parsimonious yet predictive models.

Researchers have developed two main paradigms to impose joint low-rank and sparse constraints on a matrix of interest. One approach is a \textit{decomposed model} that explicitly specifies $B = L + S$, where $L$ is low-rank and $S$ is sparse. Convex relaxations of this idea penalize $L$ and $S$ separately via the nuclear norm and the $\ell_1$ norm. Chandrasekaran et al.~(\citeyear{chandrasekaran2012latent}) decomposed covariance matrices into low-rank latent components plus sparse conditional dependencies. Farne et al.~(\citeyear{farne2023high}) used nuclear plus $\ell_1$ framework for covariance estimation, showing gains over single-structure approaches. Other examples include \cite{Zhang2017LRSRR, Bertsimas2023SLR, Tu2022LRSfMRI}.

Another approach is \textit{unified regularization}, which imposes both a nuclear norm and a sparsity-inducing norm (e.g., $\ell_1$ or $\ell_{2,1}$) directly on a single matrix $B$. Golbabaee et~al.\ (\citeyear{golbabaee2012compressed}) combined nuclear and group-sparse $\ell_{2,1}$ norms for compressed sensing of matrices that are both low-rank and jointly sparse. More recently, Lu et~al.\ (\citeyear{lu2023sparse}) proposed a unified sparse--low-rank framework for matrix quantile estimation in quadratic regression. Sparse trace-norm regularization within a single convex formulation was also investigated in \cite{Chen2014SparseTraceNorm} and \cite{Richard2012Estimation}.

Strong theoretical guarantees show that jointly enforcing sparsity and low-rankness can remain statistically robust. Bunea et~al. (\citeyear{bunea2012joint}) proved that a low-rank-plus-sparse estimator in multivariate regression achieves near-oracle prediction error. Chakraborty et~al (\citeyear{chakraborty2019bayesian}) showed that Bayesian priors inducing low-rank and row-sparse structure attain near-minimax posterior contraction without sacrificing efficiency. Ma et~al. (\citeyear{ma2020adaptive}) established near-optimal finite-sample rates for two-way (row- and column-sparse) reduced-rank regression, while Tan et~al. (\citeyear{tan2023robust}) demonstrated that adding a Huber loss preserves robustness under heavy-tailed noise.

Despite this progress, one important scenario has received limited attention, the scalar-on-matrix regression, where each input $X$ is a matrix (e.g., an image or an adjacency matrix) and the output $y$ is a scalar. Most existing sparse plus low-rank techniques focus on multi-output regression or unsupervised matrix recovery; very few methods are tailored to the single-output, matrix-input case. One of the exceptions is the \textsc{SpINNEr} method by Brzyski et~al.\ (\citeyear{brzyski2024spinner}), which imposes nuclear and $\ell_1$ penalties in a linear regression with brain connectivity matrices as predictors. 

SpINNEr optimizes a squared‑error loss for continuous outcomes, and extending it to binary responses is non‑trivial due to the non‑quadratic nature of the logistic loss. To the best of our knowledge, no existing work addresses this \emph{classification} problem in the scalar‑on‑matrix setting under joint low‑rank and sparsity constraints. We bridge this gap by proposing \textit{logistic SpINNEr}, a convex regularized method for binary scalar‑on‑matrix regression. Our approach penalizes logistic loss with both a nuclear norm and an $\ell_1$ norm on the coefficient matrix $B$, encouraging $B$ to be simultaneously low‑rank and sparse. To solve the resulting optimization problem, we develop an efficient algorithm utilizing the alternating direction method of multipliers (ADMM) combined with the iteratively reweighted least squares (IRLS), enabling scalable computation in high dimensions.

Our major contributions are as follows.
\begin{enumerate}
    \item New scalar-on-matrix classification framework, logistic SpINNEr, which incorporates joint nuclear and $\ell_1$ regularization in the coefficient matrix.
    \item Tailored ADMM-plus-IRLS optimization algorithm enabling efficient high-dimensional estimation.
    \item Synthetic experiments showing that, under the considered binary-response setting, logistic SpINNEr outperforms SpINNEr, and methods using only one structural penalty.
    \item Summary of discoveries from our real-data application to brain imaging.
\end{enumerate}

The remainder of the paper is organized as follows. Section~2 formalizes the problem setup, introduces the logistic SpINNEr model, and presents theoretical properties of the estimator under the symmetric predictors' scenario. Section~3 details the optimization procedure, outlining the ADMM algorithm and the IRLS step used to solve one of the ADMM subproblems. Section~4 presents the results of the numerical simulations. Section~5 reports the real-data analysis in a neuroimaging dataset, comparing our method with relevant baselines and illustrating how the joint structure yields interpretable patterns. Finally, Section~6 concludes with a discussion of implications and future directions.

\section{Method}
\subsection{Logistic Regression Model for Matrix Predictors}
We consider a logistic scalar-on-matrix regression model with $n$ independent observations. For each subject $i = 1, \dots, n$, let $A_i \in \mathbb{R}^{p \times p}$ denote the matrix-valued predictor and $y_i \in \{0,1\}$ the binary response. Additionally, let $X \in \mathbb{R}^{n \times d}$ be the matrix of covariates, where the $i$-th row, $X_i$, corresponds to the covariates for subject $i$. Matrix $X$ may include demographic or clinical variables such as age, sex, or disease status. An intercept is always included in the model, which is reflected by a column of ones in $X$. We define the signal matrix $B \in \mathbb{R}^{p \times p}$ and the covariate coefficient vector $\beta \in \mathbb{R}^d$. We define the linear predictor as 
\begin{equation}
    \eta_i(B, \beta) := \langle A_i, B \rangle + X_i \beta,
\end{equation} where $\langle A_i, B \rangle = \mathrm{tr}(A_i^\top B)$ denotes the Frobenius inner product. For notational simplicity, we write $\eta_i$ in place of $\eta_i(B, \beta)$ in what follows. We assume that the binary response variable $y_i$ follows a Bernoulli distribution with success probability $p_i = \big(1 + \exp(-\eta_i)\big)^{-1}$. The corresponding log-likelihood function is given by
\begin{equation}
\label{logLikelihood}
\ell(B, \beta) = \sum_{i=1}^n \big[\, y_i \eta_i - \ln\big(1 + \exp(\eta_i)\big)\, \big].
\end{equation}
In the proposed logistic regression model, $\beta$ represents the linear effects of the additional covariates on the response, while $B$ encodes the contribution of the matrix-valued predictor $A_i$ and captures structured patterns in $A_i$ that are associated with the response. We assume that $B$ is both low-rank and sparse. The low-rank assumption reflects the presence of structures that are jointly associated with the outcome, and the sparsity assumption implies that only a subset of the matrix elements is associated with the outcome.

\subsection{Logistic SpINNEr Formulation}
%
To estimate the model parameters introduced in the previous section in a way that promotes both sparsity and low-rank structure in the matrix component, we define the \textit{Logistic SpINNEr} estimator as the solution to the following optimization problem:
\begin{equation}
\label{LogisticSpinner}
\big(\hat{B}^{LS}, \hat{\beta}^{LS}\big) := \arg\min_{B,\, \beta} F(B, \beta),
\end{equation}
where the objective function is given by
\[
F(B, \beta) := -\ell(B, \beta) + \lambda_N \|B\|_* + \lambda_L \|W \circ B\|_1.
\]
Here, \( \ell(B, \beta) \) is the log-likelihood function defined in \eqref{logLikelihood}, \( \|\cdot\|_* \) denotes the nuclear norm, and \( \|\cdot\|_1 \) is the entrywise \( \ell_1 \)-norm. The term \( W \circ B \) represents the Hadamard (elementwise) product between matrices \( W \) and \( B \), where \( W \in \mathbb{R}^{p \times p} \) is a given matrix with non-negative entries. The matrix \( W \) encodes prior structural knowledge about the relative importance or penalization of individual elements in \( B \). The non-negative regularization parameters \( \lambda_N \) and \( \lambda_L \) control the balance between the low-rankness and sparsity of the solutions, respectively.
\begin{proposition}
\label{PropSolExistence}
Let $\lambda_N \geq 0$, $\lambda_L \geq 0$ and $W$ be matrix of nonnegative weights.
Define the set of two-dimensional indices
\begin{equation*}
\mathcal{I} := \left\{\, (i,j):\ \lambda_L \cdot W(i,j) + \lambda_N = 0 \,\right\}.
\end{equation*}
Assume that the data are not separable with respect to $\mathcal{I}$, i.e., there \textbf{does not exist} $(\beta, B) \in \mathbb{R}^d \times \mathbb{R}^{p \times p}$ such that $\operatorname{supp}(B) \subset \mathcal{I}$ and the following two conditions are met \newline 
(a)\ \ for all $i = 1, \ldots, n$
\begin{equation}
\label{separabilityCond}
\left\{
\begin{aligned}
y_i = 0  &\implies \langle A_i, B \rangle + X_i \beta \leq 0 \\
y_i = 1  &\implies \langle A_i, B \rangle + X_i \beta \geq 0
\end{aligned}
\right.,
\end{equation}
(b)\ \ there exists $i_0$ such that
\begin{equation}
\label{separabilityCond2}
\begin{aligned}
\big\{\,y_{i_0} = 0\ \ \textrm{and}\ \ \langle A_{i_0}, B\rangle + X_{i_0}\beta < 0\,\big\}\ \quad  \textrm{or}\\
\big\{\,y_{i_0} = 1\ \ \textrm{and}\  \ \langle A_{i_0}, B\rangle + X_{i_0}\beta > 0\,\big\}.
\end{aligned}
\end{equation}
Then, there exists a solution to the optimization problem defined in~\eqref{LogisticSpinner}.
\end{proposition}
%
\begin{proof}
    The proof is provided in Appendix~\ref{subsec:Prop21}
\end{proof}
\begin{remark}
In the special case where $\lambda_N >0$ and the covariate matrix $X$ consists of a single column of ones -- i.e., only an intercept is included in the model -- the condition that the data are not separable reduces to the requirement that the binary response vector \( y \) contains at least one observation from each class.
\end{remark}

\subsection{Symmetric Predictor Matrices}
One special case, which we give particular attention to in our work, arises when all predictor matrices \(A_i \in \mathbb{R}^{p \times p}\) are symmetric. This setting is especially relevant in neuroimaging applications, since connectivity matrices—whether structural or functional—typically represent undirected relationships between brain regions. In this context, the nonzero entries of the signal matrix \(B\) can be interpreted as weights (regression coefficients) given to the edges in an undirected graph, where each edge corresponds to a connection between a pair of brain regions.

The following two propositions, analogous to the results presented in~\cite{brzyski2024spinner}, are important for the analysis and interpretation of models within this class.
\begin{proposition}
\label{SymmetricSolution}
Assume that all predictor matrices \( A_i \in \mathbb{R}^{p \times p} \) and the weight matrix \( W \in \mathbb{R}^{p \times p} \) are symmetric. Then, the set of solutions to the optimization problem defined in~\eqref{LogisticSpinner} contains at least one symmetric matrix \( B \in \mathbb{R}^{p \times p} \).
\end{proposition}
\begin{proof}
    The proof is provided in Appendix~\ref{subsec:Prop22}.
\end{proof}
\begin{proposition}
\label{PermutingSolution}
Let $\big(\widehat{B}, \widehat{\beta}\big)$ be a solution to the optimization problem~\eqref{LogisticSpinner}. Let \( \pi : \{1, \ldots, p\} \to \{1, \ldots, p\} \) be a permutation, and let \( P_\pi \) denote the corresponding permutation matrix, so that for any column vector \( v \), \( P_\pi v = [v_{\pi(1)}, \ldots, v_{\pi(p)}]^\top \). Define the permuted predictor and weight matrices as $A_i^\pi := P_\pi A_i P_\pi^\top$ and $W^\pi := P_\pi W P_\pi^\top$. Then, solving the optimization problem with \( \{A_i^\pi\}_{i=1}^n \) and \( W^\pi \) in place of \( \{A_i\}_{i=1}^n \) and \( W \) yields a solution \( \{\widehat{B}^\pi, \widehat{\beta}\} \), where $\widehat{B}^\pi := P_\pi \widehat{B} P_\pi^\top$. In other words, permuting the rows and columns of the input matrices results in a corresponding permutation of the solution matrix \( \widehat{B} \), while the covariate coefficient \( \widehat{\beta} \) remains unchanged.
\end{proposition}
\begin{proof}
    The proof is provided in Appendix~\ref{subsec:Prop23}.
\end{proof}
\subsection{Selection of the Weight Matrix \( W \)}
The weight matrix \( W \in \mathbb{R}^{p \times p} \) controls the elementwise sparsity of the signal matrix \( B \) through the weighted \(\ell_1\)-penalty \( \|W \circ B\|_1 \). Its entries should reflect prior knowledge about the relevance of individual edges to the response. Specifically, relatively large values of \( W_{jk} \) should be assigned to edges \( (j,k) \) believed to have little or no impact on the outcome, thereby encouraging their shrinkage to zero. Conversely, smaller values of \( W_{jk} \) should be used for edges expected to be important, reducing penalization and allowing the model to retain their contribution. For edges that are structurally absent from the data—i.e., \( A_i(j,k) = 0 \) for all \( i \)—we set \( W_{jk} := 0 \) by default. This avoids unnecessary penalization of entries that cannot be informed by the data, which may help preserve the low-rank structure of \( B \). For all other edges, the default choice is \( W_{jk} := 1 \). As a result, in the context of brain imaging data, where structural or functional connectivity matrices typically have zeros on the diagonal, the default choice is to set \( W \) with ones on the off-diagonal entries and zeros on the diagonal.

\section{Optimization Algorithm}
We solve the minimization problem~\eqref{LogisticSpinner} by decomposing its objective function \( F(B, \beta) \) into three components:
\begin{equation}
F(B, \beta) = f(B, \beta) + g(B) + h(B),
\end{equation}
where \( f(B, \boldsymbol{\beta}) := -\ell(B, \boldsymbol{\beta}) \), \( g(B) := \lambda_N \|B\|_* \), and \( h(B) := \lambda_L \|\mathrm{vec}(W \circ B)\|_1 \).

Following the standard ADMM framework~\cite{boyd2011distributed}, we reformulate the problem by introducing auxiliary variables \( C \) and \( D \), and enforcing the constraints \( B = D \), \( C = D \). The next step in the ADMM procedure is to construct the augmented Lagrangian, which incorporates these constraints via quadratic penalty terms. This leads to a decomposable objective function that can be minimized by alternating updates over the primal variables  $(B, \beta)$, $C$, and $D$, followed by updates of the associated dual variables.

We initialize \( D^{[0]} \), \( U_1^{[0]} \), and \( U_2^{[0]} \) as zero matrices. Then, at each iteration \( k+1 \), we update the variables \( B^{[k+1]} \), \( C^{[k+1]} \), \( D^{[k+1]} \) by solving the following three subproblems:
\begin{align}
\big(B^{[k+1]}&,\beta^{[k+1]}\big):= \notag\\
&\argmin{B, \beta} \left\{ 2f(B,\beta) + \delta_1^{[k]} \|D^{[k]} + \tfrac{U_1^{[k]}}{\delta_1^{[k]}} - B\|_F^2 \right\} , \tag{SP 1} \label{eq:update_B_beta} \\
C^{[k+1]}& := \argmin{C} 
\left\{ 2g(C) + \delta_2^{[k]} \|D^{[k]} + \tfrac{U_2^{[k]}}{\delta_2^{[k]}} - C\|_F^2 \right\}, \tag{SP 2} \label{eq:update_C} \\
D^{[k+1]} &:= \argmin{D} \left\{ 2h(D) + \delta_1^{[k]} \|D + \tfrac{U_1^{[k]}}{\delta_1^{[k]}} - B^{[k+1]}\|_F^2 \right. \notag \\
&\quad \left\{ + \tfrac{\delta_2^{[k]}}{2} \|D + \tfrac{U_2^{[k]}}{\delta_2^{[k]}} - C^{[k+1]}\|_F^2 \right\}, \tag{SP 3} \label{eq:update_D}
\end{align}
The multiplier matrices, \( U_1^{[k+1]} \), and \( U_2^{[k+1]} \),  are updated as 
\begin{equation}
\label{eq:update_U}
\begin{aligned}
U_1^{[k+1]} := U_1^{[k]} + \delta_1^{[k]} (D^{[k+1]} - B^{[k+1]}), \\
U_2^{[k+1]} := U_2^{[k]} + \delta_2^{[k]} (D^{[k+1]} - C^{[k+1]}).
\end{aligned}
\end{equation}

This formulation separates the smooth (log-likelihood) and non-smooth (nuclear and \(\ell_1\)) components, enabling efficient updates via proximal operators. The use of auxiliary variables \( C \) and \( D \) allows decoupling of the low-rank and sparsity-inducing penalties, while the dual variables \( U_1 \) and \( U_2 \) enforce consensus among the variables.

Subproblems~\eqref{eq:update_C} and~\eqref{eq:update_D} admit closed-form solutions, as introduced in~\cite{brzyski2024spinner}. In contrast, subproblem~\eqref{eq:update_B_beta} lacks an analytical solution and requires an iterative numerical solver. Designing an efficient algorithm for solving~\eqref{eq:update_B_beta} represents the primary computational challenge in the implementation. A detailed description of the employed solution strategy is provided in Section~\ref{subsect:SolvingSP1}.

The iterative procedure—comprising the updates of \( B \), \( C \), and \( D \) via subproblems~\eqref{eq:update_B_beta}–\eqref{eq:update_D}, followed by the multiplier updates in~\eqref{eq:update_U}—is repeated until convergence. The stopping criterion is based on the relative Frobenius norms of the primal and dual residuals. Specifically, the algorithm terminates when the following conditions are satisfied: \( \|D^{[k+1]} - B^{[k+1]}\|_F / \|B^{[k+1]}\|_F < \varepsilon_{\mathrm{pri}} \), \( \|D^{[k+1]} - C^{[k+1]}\|_F / \|B^{[k+1]}\|_F < \varepsilon_{\mathrm{pri}} \), and \( \|D^{[k+1]} - D^{[k]}\|_F / \|D^{[k]}\|_F < \varepsilon_{\mathrm{dual}} \), where \( \varepsilon_{\mathrm{pri}} \) and \( \varepsilon_{\mathrm{dual}} \) are predefined tolerance levels. 

To balance the convergence rates of all terms involved in the stopping criterion, the penalty parameters \( \delta_1^{[k]} \) and \( \delta_2^{[k]} \) are adaptively updated every given number of iterations (10 by default). Specifically, the ratio \( \delta_1^{[k]} / \delta_2^{[k]} \) is adjusted based on a comparison between \( \|D^{[k+1]} - B^{[k+1]}\|_F \) and \( \|D^{[k+1]} - C^{[k+1]}\|_F \), while the overall scaling factor is modified according to the relative magnitudes of the primal and dual residuals. These adaptive updates help maintain balanced and stable convergence behavior.

\subsection{Optimization Strategy for \ref{eq:update_B_beta}}
\label{subsect:SolvingSP1}
After vectorizing the matrices \( A_i \) into row vectors \( {\mathcal{A}_i := \mathrm{vec}(A_i)^\top} \), we stack them to form the matrix \( \mathcal{A} \in \mathbb{R}^{n \times p^2} \) and define $Z \in \mathbb{R}^{n \times (p^2 + d)}$ as \( Z := [\mathcal{A}\, X]  \). Now, with  $\theta := [\mathrm{vec}(B)^\top \, \beta^\top]^\top \in \mathbb{R}^{p^2 + d}$, \eqref{eq:update_B_beta} can be rewritten as a minimization problem over \( \theta \) with the objective function
\begin{equation}
\mathcal{L}(\theta) := \sum_{i=1}^n \left[ \ln\left(1 + e^{\eta_i}\right) - y_i \eta_i \right] + \frac{\delta}{2} \left\| P (\theta - \gamma ) \right\|_2^2,
\end{equation}
where  
\( \eta_i = Z_i \theta \),\, \( \delta:= \delta_1^{[k]} \),\, \( P:= \mathrm{diag}(\mathrm{I}_{p^2},\, 0_{d}) \),\, and\,  
\( \gamma:= [\mathrm{vec}(Y)^\top ,\, 0_d^\top]^\top \),\, with \( Y:= D^{[k]} + \tfrac{U_1^{[k]}}{\delta_1^{[k]}} \).

This formulation corresponds to a regularized logistic regression problem that can be solved using the \emph{Iteratively Reweighted Least Squares (IRLS)} method. In this context, IRLS is implemented via the Newton–Raphson algorithm, which iteratively updates \( \theta \) using the gradient \( \nabla \mathcal{L}(\theta) \) and the Hessian \( H(\theta) \) of the objective function. Define \( \mu \in \mathbb{R}^n \) by setting \( \mu_i := 1/(1 + e^{-\eta_i}) \) for \( i = 1, \dots, n \), and define the diagonal matrix \( \mathcal{D} \in \mathbb{R}^{n \times n} \) with diagonal entries \( \mathcal{D}_{ii} := \mu_i (1 - \mu_i) \). Then, the gradient of \( \mathcal{L}(\theta) \) is given by
\begin{equation}
\nabla \mathcal{L}(\theta) = Z^\top (\mu - y) + \delta P (\theta - \gamma),
\end{equation}
while the Hessian becomes
\begin{equation}
\label{eq:Hessian}
H(\theta) = Z^\top \mathcal{D} Z + \delta P.
\end{equation}
The Newton–Raphson update rule is given by
\begin{equation}
\label{NRrule}
\theta^{(t+1)} = \theta^{(t)} - H\big(\theta^{(t)}\big)^{-1} \nabla \mathcal{L}\big(\theta^{(t)}\big).
\end{equation}

Rewriting the gradient at \( \theta^{(t)} \) we get
\begin{equation}
\label{eq:gradient2}
  \begin{aligned}
    \nabla \mathcal{L}&(\theta^{(t)}) =\\ 
      & Z^\top (\mu^{(t)} - y) + \delta P (\theta^{(t)} - \gamma ) =\\
			& \left(Z^\top \mathcal{D}^{(t)} Z + \delta P \right) \theta^{(t)} + Z^\top \mathcal{D}^{(t)} z^{(t)} - \delta P \gamma,
  \end{aligned}
\end{equation}
for \( z^{(t)} := \big(\mathcal{D}^{(t)}\big)^{-1}(\mu^{(t)} - y) - \eta^{(t)} \).

Substituting \eqref{eq:Hessian} and \eqref{eq:gradient2} to~\eqref{NRrule} we get
\begin{equation}
\label{eq:theta_tPlus1}
\theta^{(t+1)} = \left(Z^\top \mathcal{D}^{(t)} Z + \delta P \right)^{-1} \big(Z^\top \mathcal{D}^{(t)} z^{(t)} + \delta P \gamma  \big)
\end{equation}

After multiplying both sides of \eqref{eq:theta_tPlus1} by $(Z^\top \mathcal{D}^{(t)} Z + \delta P)$ and solving the corresponding equation for $\mathrm{vec}\big(B^{(t+1)}\big)$ and $\beta^{(t+1)}$ we finally get
\begin{equation}
\label{UpdateFormulas}
\left\{
\begin{array}{l}
\begin{aligned}
\mathrm{vec}\big(B&^{(t+1)}\big)= \\
&\big(\mathcal{A}^\top \widetilde{\mathcal{D}} \mathcal{A} + \delta I_{p^2}\big)^{-1}\big(\mathcal{A}^\top \widetilde{\mathcal{D}}  z + \delta\, \mathrm{vec}(Y)\big)\\
\beta^{(t+1)} &= \big(X^\top \mathcal{D}X\big)^{-1} X^\top \mathcal{D} \Big(z - \mathcal{A} \mathrm{vec}\big(B^{(t+1)}\big)\Big),
\end{aligned}
\end{array}
\right.
\end{equation}
for $\widetilde{\mathcal{D}}:= \mathcal{D}-\mathcal{D} X\Big(X^\top \mathcal{D} X\Big)^{-1}X^\top \mathcal{D}$. In the above update formulas, only the quantities $\mathcal{D}$, $\widetilde{\mathcal{D}}$, and $z$ depend on the iteration $t$ (for simplicity, we omit the superscript notation).

\subsection{Efficient Implementation via Reduced SVD}
While the analytically convenient update formula \eqref{UpdateFormulas} provides a direct solution for $(B,\beta)$ at each iteration, it requires inverting a $p^2 \times p^2$ matrix. In many applications, $p^2 \gg n$, making the direct computation of $B^{(t+1)}$ computationally expensive. To overcome this issue in the case $p^2 > n$, we reformulate the update step in a way that involves only an $n \times n$ matrix inversion, leading to a significant speed-up. The key idea is to exploit the \emph{reduced singular value decomposition (SVD)} of $\mathcal{A}$:
\begin{equation}
\label{reducedSVD}
\mathcal{A} = U [\, S \ \ \bm{0}_{n\times (p^2-n)}\, ]\ [ V_1 \ \  V_2 ]^\top,
\end{equation}
where \( S \in \mathbb{R}^{n \times n} \) is diagonal matrix, $V_1 \in \mathbb{R}^{p^2 \times n}$ and $V_2\in \mathbb{R}^{p^2 \times (p^2-n)}$ are submatrices of orthonormal matrix \( V :=[ V_1\ V_2 ]\). 

This decomposition needs to be computed once for the entire ADMM algorithm and, importantly, we only need to retain the $n$ right-singular vectors corresponding to the nonzero singular values.

Let $\tilde{S}: = [\, S \ \ \bm{0}_{n\times (p^2-n)}\, ]$. Thanks to \eqref{reducedSVD}, $B$ update in \eqref{UpdateFormulas} can be rewritten as
\begin{equation}
\mathrm{vec}\big(B^{(t+1)}\big) = V \mathcal{M}\, \omega,
\end{equation}
where
\begin{equation}
\begin{aligned}
\mathcal{M}: =\ \Big( \tilde{S}^\top & U^\top \widetilde{\mathcal{D}}  U \tilde{S}  + \delta I_{p^2} \Big)^{-1} = \\
&\begin{bmatrix}
(S U^\top \widetilde{\mathcal{D}} U S + \delta I_n)^{-1} & 0 \\
0 & \delta^{-1} I_{p^2-n}
\end{bmatrix}
\end{aligned}
\end{equation}
and
\begin{equation}
\begin{aligned}
\omega: =\ \tilde{S}^\top  U^\top& \widetilde{\mathcal{D}}\, z  \, +\,  \delta\, V^\top \mathrm{vec}(Y) = \\
&\begin{bmatrix}
\tilde{S}^\top U^\top \widetilde{\mathcal{D}}\, z\, +\, \delta\, V_1^\top \mathrm{vec}(Y) \\
\delta\, V_2^\top \mathrm{vec}(Y)
\end{bmatrix}.
\end{aligned}
\end{equation}
Substituting the above, we obtain
\begin{equation}
\begin{aligned}
\mathrm{vec}\big(B^{(t+1)}\big) & = \ \ V_1 \widetilde{z} + V_2 V_2^\top \mathrm{vec}(Y)\ =\\
&V_1 \widetilde{z} - V_1 \nu_1 + \mathrm{vec}(Y),
\end{aligned}
\end{equation}
where $\nu_1 := V_1^\top \mathrm{vec}(Y)$ and
\begin{equation*}
\widetilde{z}: = \big(S U^\top \widetilde{\mathcal{D}} U S + \delta I_n\big)^{-1}\big(S U^\top \widetilde{\mathcal{D}} z + \delta \nu_1\big),
\end{equation*}
since $V_1 V_1^\top + V_2 V_2^\top = I$ by the orthonormality of $V$.
We also have
\begin{equation}
\mathcal{A}\mathrm{vec}\big(B^{(t+1)}\big)  = US\widetilde{z}
\end{equation}
and the entire procedure is summarized in \cref{alg:svd_irls}.
\begin{algorithm}[tb]
\caption{SVD-Based algorithm for IRLS step}
\label{alg:svd_irls}
\begin{algorithmic}[1]
\STATE \textbf{Input:} Singular value matrix \( S \in \mathbb{R}^{n \times n} \) and left singular vectors \( U \in \mathbb{R}^{n \times n} \) of \(\mathcal{A}\), covariate matrix \( X \in \mathbb{R}^{n \times d} \), binary response vector \( y \in \{0,1\}^n \), step-size parameter \( \delta \), matrix \( Y \in \mathbb{R}^{p \times p} \), auxiliary vectors \( \nu_1 := V_1^\top \mathrm{vec}(Y) \) and \( \nu_2 := \mathrm{vec}(Y) - V_1 \nu_1 \)
\STATE \textbf{Initialize:} \( \mathrm{vec}(B^{(0)}), \beta^{(0)}, \widetilde{z} \leftarrow 0 \)
\FOR{each iteration \( t = 0, 1, 2, \dots \) until convergence}
    \STATE \( \eta \leftarrow US\widetilde{z} + X \boldsymbol{\beta}^{(t)} \)
    \STATE Define \( \mu \in \mathbb{R}^n \) by \( \mu_i \leftarrow \frac{1}{1 + e^{-\eta_i}} \) for \( i = 1, \dots, n \)
    \STATE Define diagonal matrix \( \mathcal{D} \in \mathbb{R}^{n \times n} \) with \( \mathcal{D}_{ii} \leftarrow \mu_i (1 - \mu_i) \)
    \STATE \( z \leftarrow \eta - \mathcal{D}^{-1}(\mu - y) \)
    \STATE \( \widetilde{\mathcal{D}} \leftarrow \mathcal{D} - \mathcal{D} X (X^\top \mathcal{D} X)^{-1} X^\top \mathcal{D} \)
    \STATE \( \widetilde{z} \leftarrow (S U^\top \widetilde{\mathcal{D}} U S + \delta I_n)^{-1} (S U^\top \widetilde{\mathcal{D}} z + \delta \nu_1) \)
    \STATE \( \mathrm{vec}(B^{(t+1)}) \leftarrow V_1 \widetilde{z} + \nu_2 \)
    \STATE \( \beta^{(t+1)} \leftarrow (X^\top \mathcal{D} X)^{-1} X^\top \mathcal{D} (z - US\widetilde{z}) \)
\ENDFOR
\STATE \textbf{Output:} \( \mathrm{vec}(B^{(t+1)}), \boldsymbol{\beta}^{(t+1)} \)
\end{algorithmic}
\end{algorithm}
The matrix inversions required in \cref{alg:svd_irls} appear in lines 7, 8, 9, and 11. Among these, the inversions in lines 8 and 11 involve only \( d \times d \) matrices and are therefore computationally inexpensive. Similarly, the inversion in line 7 is trivial, as it concerns a diagonal matrix. The primary computational bottleneck arises in line 9. However, this only requires inverting an \( n \times n \) matrix and it is significantly more efficient than in the update formulation \eqref{UpdateFormulas}. Moreover, the current estimates $\big(B^{[k]}, \beta^{[k]}\big)$ obtained from the main ADMM loop (subproblems~\eqref{eq:update_B_beta}--\eqref{eq:update_D} and the dual update step~\eqref{eq:update_U}) can be effectively used as initialization for \( \mathrm{vec}(B^{(0)}) \) and \( \beta^{(0)} \). This warm-start strategy may accelerate convergence and improve numerical stability.

\section{Experimental Results}
\subsection{Setup}
We conducted a simulation study to evaluate the performance of the Logistic SpINNEr and compare it to three alternative methods: \textit{SpINNEr}, introduced in~\cite{brzyski2024spinner}; \textit{Logistic LASSO}, a single-penalty regularization approach with objective function $-\ell(B, \beta) + \lambda \|W \circ B\|_1$; and \textit{Logistic NUCLEAR}, which applies a nuclear norm penalty via the objective function $-\ell(B, \beta) + \lambda \|B\|_*$, where $\ell(B, \beta)$ is defined in~\cref{logLikelihood}.

Two simulation scenarios were considered to assess the impact of sample size and signal strength on the estimation performance. In the first scenario, the number of observations $n$ varied over a grid $n \in \{100, 120, \dots, 240\}$, while the true signal remained fixed. In the second scenario, the sample size was fixed at $n = 300$, and the true signal amplitude was scaled by a factor $\gamma \in \{0.1, 0.25, 0.5, 1, 1.5, 2\}$.

For each grid value (either sample size or signal strength), we generated $50$ independent datasets. Each dataset consisted of $n$ observations with symmetric matrix-valued predictors $\mathbf{A}_i \in \mathbb{R}^{60 \times 60}$, whose entries were sampled from a standard normal distribution and symmetrized. A constant intercept $x_i = 1$ was included in all cases.

The true coefficient matrix $\mathbf{B}_{\text{tr}}$ contained three nonzero blocks—two with positive entries and one with negative entries—resulting in a total of $66$ nonzero elements (see~\cref{fig:true-signal}). The intercept coefficient was set to $\mathbf{\beta}_{\text{tr}} = 1$. Responses were generated according to the logistic model ${y_i \sim \mathrm{Bernoulli}\big( 1/(1 + e^{-\eta_i}) \big)}$, for $\eta_i = \langle A_i, \mathbf{B}_{\text{tr}} \rangle + \mathbf{\beta}_{\text{tr}}$.

The regularization parameters for all methods were selected via 5-fold cross-validation: over $10$ values for single-penalty methods, and over a $10 \times 10$ grid of $(\lambda_L, \lambda_N)$ values for Logistic SpINNEr and SpINNEr. Estimation accuracy was assessed using the \textit{Relative Frobenius Error}, defined as $\mathbb{E}\left[\|\hat{\mathbf{B}} - \mathbf{B}_{\text{tr}}\|_F/\|\mathbf{B}_{\text{tr}}\|_F\right]$. For each configuration, the reported value corresponds to the empirical mean over $50$ Monte‑Carlo repetitions.
\subsection{Results}
\begin{figure*}[ht]
  \centering
  \begin{subfigure}[t]{0.405\columnwidth}
    \centering
    \includegraphics[width=\linewidth]{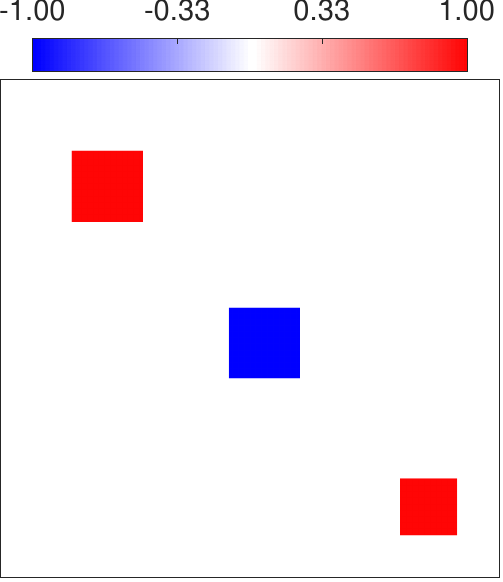}
    \caption{True signal}
    \label{fig:true-signal}
  \end{subfigure}
  \begin{subfigure}[t]{0.405\columnwidth}
    \centering
    \includegraphics[width=\linewidth]{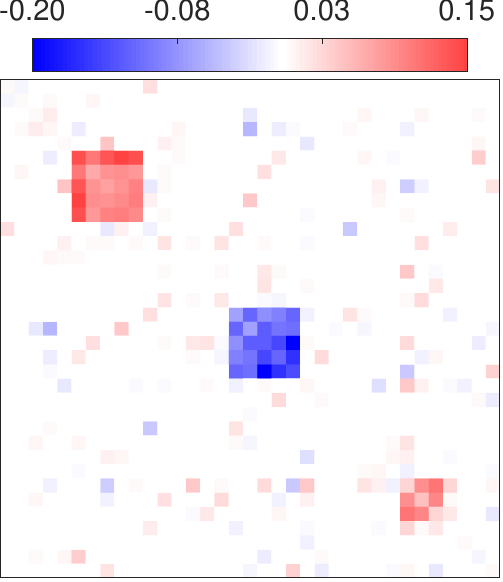}
    \caption{Logistic SpINNEr}
    \label{fig:heatmap-logsp}
  \end{subfigure}
  \begin{subfigure}[t]{0.405\columnwidth}
    \centering
    \includegraphics[width=\linewidth]{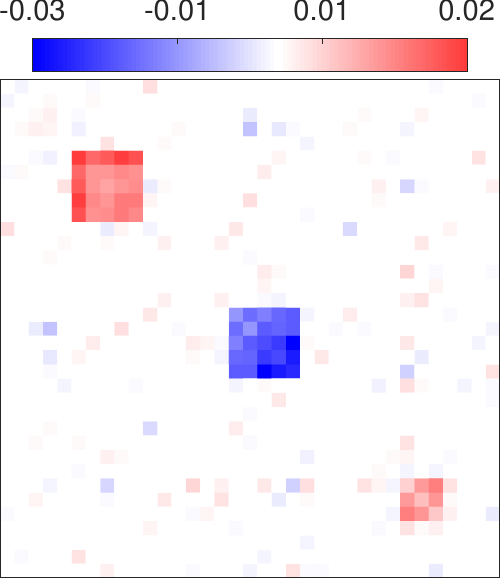}
    \caption{SpINNEr}
    \label{fig:heatmap-sp}
  \end{subfigure}
  \begin{subfigure}[t]{0.405\columnwidth}
    \centering
    \includegraphics[width=\linewidth]{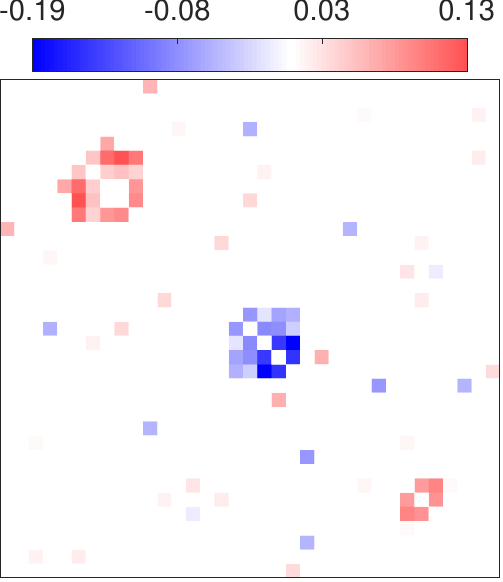}
    \caption{Logistic LASSO}
    \label{fig:heatmap-logla}
  \end{subfigure}
  \begin{subfigure}[t]{0.405\columnwidth}
    \centering
    \includegraphics[width=\linewidth]{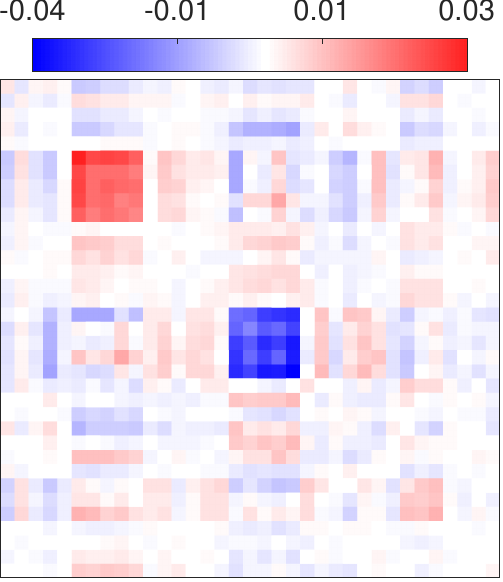}
    \caption{Logistic Nuclear}
    \label{fig:heatmap-lognu}
  \end{subfigure}
  \begin{subfigure}[t]{0.99\columnwidth}
    \centering
    \includegraphics[width=\linewidth]{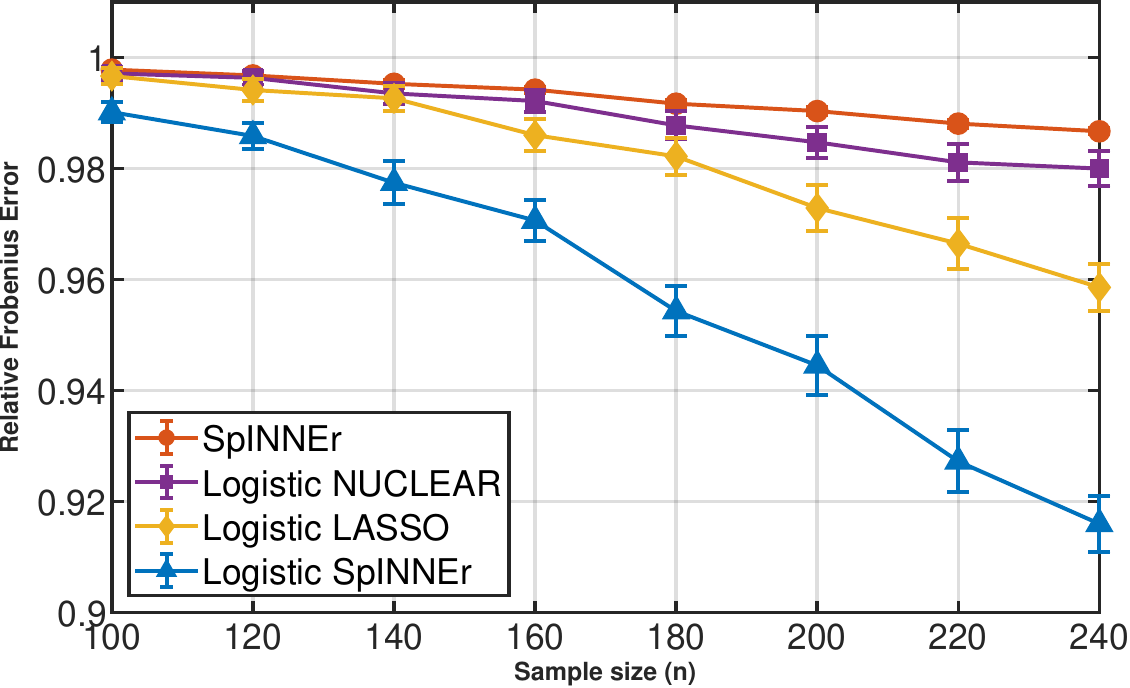}
    \caption{Estimation error for increasing sample size}
    \label{fig:perf-samplesize}
  \end{subfigure}
  \quad
  \begin{subfigure}[t]{0.99\columnwidth}
    \centering
    \includegraphics[width=\linewidth]{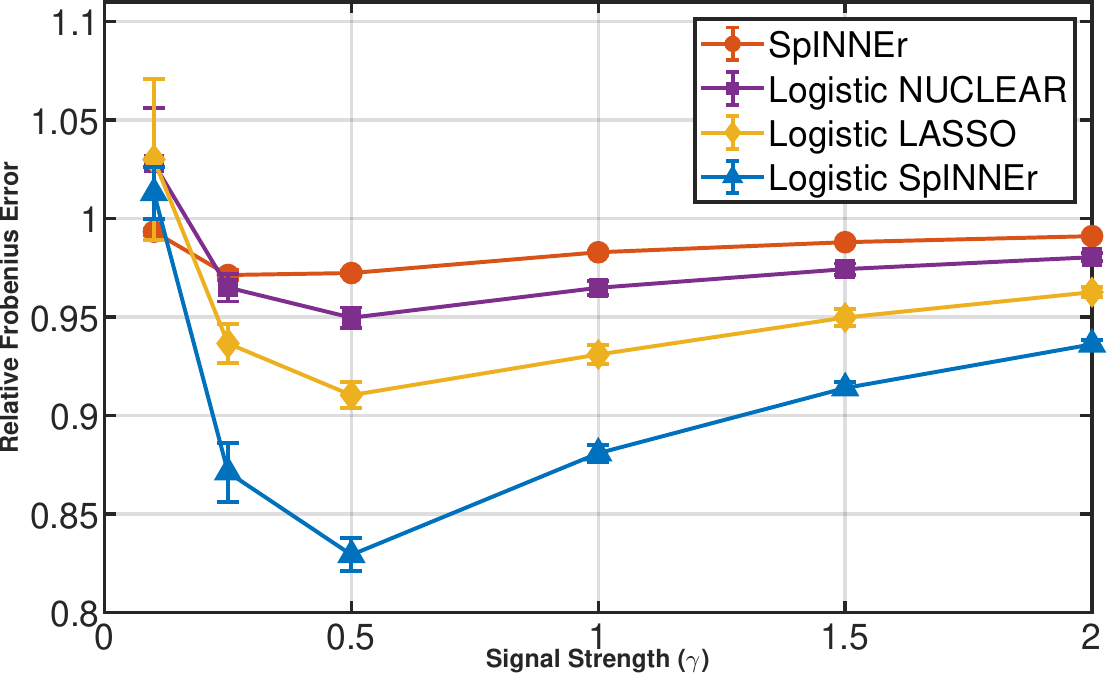}
    \caption{Estimation error for increasing signal strength}
    \label{fig:perf-signal}
  \end{subfigure}
  \caption{Panel~(a) displays the heatmap of the true coefficient matrix, while panels~(b)--(e) show the estimated coefficient matrices obtained from the first simulation replicate for each of the four methods (for $n = 300$ and $\gamma = 0.5$). For clarity, only the upper-left $35 \times 35$ submatrix is shown, corresponding to the region containing all nonzero blocks. Among the single-penalty methods, the $\ell_1$ norm, (d), yields better inter-block sparsity, whereas the nuclear norm, (e), achieves superior estimation within the true blocks but allows more false positives between blocks. The joint use of both penalties, (c) and~(b), results in a cumulative effect, balancing intra- and inter-block accuracy. Panels~(f) and~(g) display the average estimation error, defined as $\|\hat{\mathbf{B}} - \mathbf{B}_{\text{tr}}\|_F/\|\mathbf{B}_{\text{tr}}\|_F$, computed across $50$ repetitions. Panel~(f) illustrates how the error drops for all four methods as the sample size $n$ increases. Panel~(g) shows the effect of signal strength: the error initially decreases as the true effect size grows, reaches a minimum at moderate signal levels, and then increases again for very strong signals. Error bars represent $\pm 2$ standard errors.}
   \label{fig:heatmaps}
\end{figure*}

Figure~\ref{fig:heatmaps} shows heatmaps of the estimated coefficient matrices from all four methods alongside the true signal matrix (only the upper-left $35 \times 35$ submatrix containing the three nonzero blocks is displayed). Both Logistic SpINNEr and SpINNEr successfully recover the block-sparse structure of the true signal, exhibiting similar patterns of nonzero edges (see subfigures~(b) and~(c)). However, Logistic SpINNEr has substantially less shrinkage of the estimated coefficients toward zero, resulting in more accurate recovery of the true signal magnitudes. The single-penalty methods behave as expected. Namely, Logistic LASSO produces a sparser estimate, often missing several significant edges within the true blocks. In contrast, Logistic NUCLEAR achieves strong intra-block estimation accuracy but tends to introduce numerous spurious inter-block non-zero coefficients, reflecting its inability to enforce entry-wise sparsity effectively.

Figure~\ref{fig:perf-samplesize} shows the average relative Frobenius error as a function of the sample size. While all methods exhibit improved accuracy with increasing $n$, Logistic SpINNEr consistently achieves the lowest error across the entire sample size range. In contrast, SpINNEr yields the highest estimation error, indicating that an inappropriate model likelihood—even when paired with suitable regularization—can lead to substantial overshrinkage and degraded estimation performance.

Figure~\ref{fig:perf-signal} illustrates the behavior of the methods with the increasing signal strength. The observed, non-monotonic shape of the relative estimation error in regularized regression is a known phenomenon: when the signal is weak, regularization tends to suppress it as noise; as signal strength increases, estimation accuracy improves sharply; and eventually, the error may stabilize or slightly increase due to persistent regularization bias \cite{van2008high, zhang2014confidence, vazquez2025debiased, sur2019modern}.

In summary, Logistic SpINNEr outperforms the other methods across both simulation scenarios. Its ability to enforce sparsity and low-rank structure within a logistic regression framework enables it to recover structured signals more accurately, particularly in moderate signal-strength settings.

\section{Real Data Analysis}
We applied the proposed logistic SpINNEr regression framework to task-based fMRI data collected as part of a study on genetic risk for alcohol use disorder (AUD). Prior work already demonstrated altered structural and functional connectivity in individuals with familial AUD risk \cite{wetherill2012, weiland2013, herting2011}.

fMRI data were acquired while participants received brief spritzes of a high-sucrose solution, a paradigm designed to engage reward-related neural circuitry. The sample consisted of $n=161$ young adults exhibiting either social drinking patterns or increased alcohol use \cite{alessi2024high}.

Preprocessed fMRI time series were parcellated into 200 cortical regions of interest (ROIs) using the Schaefer atlas \cite{schaefer2018local} , with network assignments based on the Yeo 7-network organization \cite{yeo2011organization}. For each subject, a 200 × 200 functional connectivity matrix was constructed using pairwise correlations between the ROI-specific time series. These subject-specific connectivity matrices served as the sole predictors in a binary regression model, with family history of alcoholism (family-history-positive vs. family-history-negative) as the binary response.

To assess stability and mitigate sensitivity to sampling variability, we employed a bootstrap aggregation procedure \cite{efron1994introduction}. Specifically, 300 bootstrap samples were drawn at the subject level, and the penalized model was refit to each sample, yielding 300 estimated coefficient matrices. These estimates were aggregated by retaining only those matrix entries that were estimated to be nonzero in at least 90\% of the bootstrap fits. For the retained entries, the final coefficient values were taken as the median across bootstrap replications.
The aggregated coefficient matrix obtained from the bootstrap procedure contained 41 nonzero entries, corresponding to functional connections associated with family history of alcoholism. Both positive and negative associations were identified, reflecting increased or decreased connectivity with respect to genetic risk. Given the high dimensionality of the predictor space, this small number of retained edges highlights the strong regularization and stability properties of the proposed scalar-on-matrix penalized regression approach.
Figure~\ref{fig:glass-brain} provides an anatomical visualization of the 41 selected connections using a glass-brain representation generated with the BrainNet Viewer software \cite{xia2013brainnet}. The figure consists of four panels arranged in a single row showing the right and left lateral as well as axial and coronal viewpoints. This visualization highlights the spatial distribution of selected edges across cortical surfaces, enabling assessment of laterality and gross anatomical organization without over-interpretation.
\begin{figure*}[!tb]
  \centering
  \includegraphics[
    width=\textwidth,
    keepaspectratio
  ]{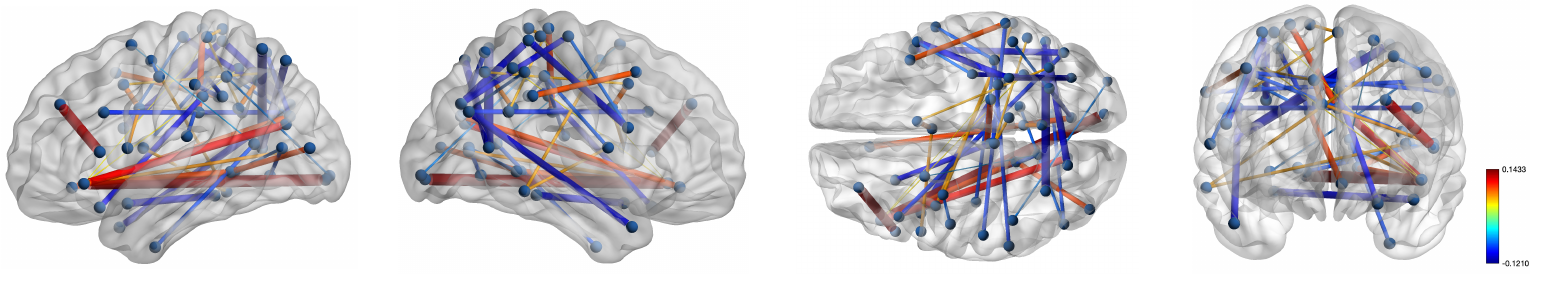}
  \caption{Glass brain network visualization. From left to right: right lateral, left lateral, axial (superior), and coronal (posterior) viewpoints.}
  \label{fig:glass-brain}
\end{figure*}
Figure~\ref{fig:circ-network} displays the estimated coefficient matrix as a circular network visualization, in which the 200 cortical regions are grouped according to the Yeo 7-network organization. Regions belonging to the same functional network are arranged contiguously along seven concentric arcs, which are themselves positioned in a circular layout. Edges are color-coded to reflect both the sign of the association (calculated via the average value across the bootstraps for a given edge) and by whether the connection occurs within or between networks. This representation emphasizes the network-level organization of the selected connections and illustrates that the identified effects are distributed across multiple large-scale functional systems rather than being confined to a single network or an isolated region. Together with the glass-brain representation, the circular network visualization demonstrates that the proposed method yields sparse yet anatomically coherent results, while preserving a direct correspondence between estimated coefficients and interpretable functional connections.
\begin{figure}[ht]
  \begin{center}
    \centerline{\includegraphics[width=1\columnwidth]{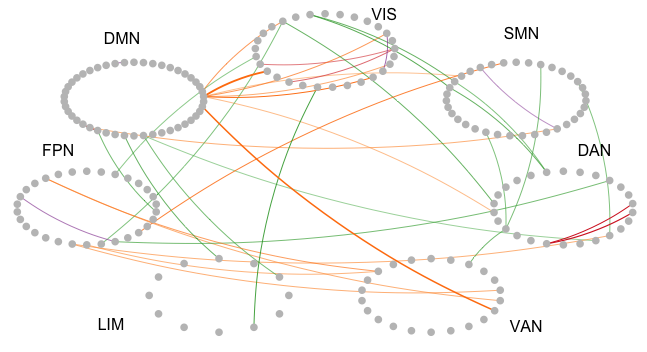}}
    \caption{
      Network graph showing edges present in 90 percent of the bootstrap samples. The labeled networks are visual (VIS), somato-motor (SMN), dorsal attention (DAN), ventral attention (VAN), limbic (LIM), fronto-parietal (FPN) and Default Mode Network (DMN). Edges colors represent mean positive (P) vs. negative (N) association across the bootstraps and within (W) vs. between (B) Yeo subnetworks: purple for P/W network, orange for P/B, red for N/W, and green for N/B.
    }
    \label{fig:circ-network}
  \end{center}
\end{figure}
\section{Discussion}
Our results demonstrate that employing a logistic loss function tailored to binary response data can significantly enhance estimation accuracy. Furthermore, we confirm that imposing joint nuclear and $\ell_1$ penalties in scalar-on-matrix regression offers clear advantages over single-structure regularization, when the underlying signal is both sparse and low-rank. In both simulation scenarios, Logistic SpINNEr consistently outperformed alternative methods, achieving lower estimation error.

To solve the optimization problem for Logistic SpINNEr, we develop an efficient algorithm based on the Alternating Direction Method of Multipliers (ADMM). Moreover, by incorporating a singular value decomposition within the IRLS step used to solve one of the ADMM subproblems, we replace a costly operation on $p^2 \times p^2$ matrices with an equivalent operation on $n \times n$ matrices, yielding a substantial reduction in overall computation time.

In the neuroimaging application, our findings underscore the practical impact of the proposed method for discovering neurobiologically meaningful patterns. When applied to the functional connectivity AUD dataset, logistic SpINNEr identified a sparse set of 41 connections (0.21\% out of 19,800 possible) that distinguished individuals with a family history of AUD. Notably, these selected connectivity edges were distributed across multiple large-scale functional systems rather than confined to a single network, indicating that the familial risk signature involves a broad configuration of connectivity changes. Both within-network and between-network connections emerged, suggesting that our joint regularization approach can uncover complex, distributed biomarkers that simpler models might overlook. 

Despite its strengths, logistic SpINNEr has also certain limitations. First, the added flexibility of dual penalties comes with increased computational complexity and the need to tune two regularization parameters $\lambda_N$ and $\lambda_L$. While our ADMM-based algorithm (with dimensionality reduction techniques and warm starts) mitigates the computational cost, training is still more intensive than using a single-penalty model. In large-scale settings (e.g., extremely high-resolution images or networks), further efficiency improvements or distributed computing might be required. Second, the model assumes that the relevant signal can be well-approximated by a combination of sparse and low-rank structure, but if the true underlying pattern does not conform to these structures, estimation performance may degrade or the interpretation may become less clear.

An important avenue for future research involves the development of more efficient and theoretically grounded strategies for selecting the regularization parameters $\lambda_N$ and $\lambda_L$. While our current implementation relies on grid-based cross-validation, this approach can be computationally intensive, particularly in high-dimensional settings. Adaptive or data-driven tuning procedures, such as empirical Bayes methods or information criteria tailored to structured penalties, may offer more scalable and statistically robust alternatives.



\section*{Impact Statement}
Our work introduced a method for analyzing high-dimensional structured data. By improving detection of meaningful patterns in neuroimaging data, it can advance research on complex biological systems and other areas requiring structured classification. Potential societal benefits include deeper insights into brain connectivity mechanisms, while risks concern inappropriate use of data-driven techniques in healthcare and related areas. We encourage responsible application and transparency when applying these methods.

\bibliography{references}
\bibliographystyle{icml2026}

\newpage
\appendix
\onecolumn
\section{Proofs of Logistic SpINNEr Properties}
\subsection{Proof of Proposition~\ref{PropSolExistence}}
\label{subsec:Prop21}
We will show that all the directions of recession of $F$ are directions of constancy, which implies that $F$ attains its minimum value \cite{rockafellar1970convex, bertsekas2009convex}. Let $r: = (B, \beta) \in \mathbb{R}^{p\times p} \times \mathbb{R}^d$ be any vector such that $r\neq 0$. We begin with the case where $\textrm{supp}(B) \subset \mathcal{I}$. Without loss of generality, we exclude the case $\{B \neq 0\ \textrm{and} \ \lambda_N > 0 \}$, because the positiveness of $\lambda_N$ implies that $\mathcal{I} = \emptyset$. Consequently, $B = 0$\ or $\lambda_N = 0$ and we get $\lambda_N\|B\|_* + \lambda_L\|W\circ B\|_1 = 0$. This reduces the situation to the logistic regression and gives
\begin{equation}
\lim_{c\rightarrow +\infty} F\big(\,c(B, \beta)\,\big) = \lim_{c\rightarrow +\infty} \sum_{i=1}^n f_{c,i}\big(\,\langle A_i, B\rangle + X_i\beta\big),\qquad \textrm{for} \quad  f_{c,i}(x): = \ln\big(\,1+exp(cx)\,\big) - cy_ix.
\end{equation}
Now, for all $i=1,\ldots,n$,\ it holds \ \  $f_{c,i}(x) \geq 0 $\ for all $x\in\mathbb{R}$ and
\begin{equation*}
\big\{\,y_i = 0\,\big\} \implies  \lim_{c \rightarrow +\infty} f_{c,i}(x) = 
\begin{cases}
0 \quad \quad \ ,\ \textrm{if}\ \ x < 0 \\
\ln(2)\ , \  \textrm{if}\ \ x = 0 \\
\infty\qquad, \textrm{if}\ \  x > 0 \\
\end{cases}, \quad \ \
\big\{\,y_i = 1\,\big\}\implies \lim_{c \rightarrow +\infty} f_{c,i}(x) = 
\begin{cases}
0 \quad \quad \ ,\ \textrm{if}\ \ x > 0 \\
\ln(2)\ , \  \textrm{if}\ \ x = 0 \\
\infty\qquad, \textrm{if}\ \  x < 0 \\
\end{cases}.
\end{equation*}
Suppose, for contradiction, that $r$ is a direction in which $F$ decreases. Then for all $i=1,\ldots,n$, we must have
$$\big\{\,y_i = 0\implies\langle A_i, B\rangle + X_i\beta \leq 0\, \big\}\qquad \textrm{and}\qquad \big\{\,y_i = 1\implies \langle A_i, B\rangle + X_i\beta \geq 0\, \big\}$$
and there exists $i_0 \in \{1,\ldots,n\}$ such that 
$$\big\{\,y_{i_0} = 0\ \textrm{and} \ \langle A_{i_0}, B\rangle + X_{i_0}\beta < 0\, \big\}\qquad \textrm{or}\qquad \big\{\,y_{i_0} = 1\ \textrm{and} \ \langle A_{i_0}, B\rangle + X_{i_0}\beta > 0\, \big\}.$$ 
This contradicts the assumption of the absence of the separability conditions~{\eqref{separabilityCond}--\eqref{separabilityCond2}}.

Now consider the case when there exists $(i,j) \notin \mathcal{I}$ such that $B(i,j) \neq 0$. Since $-\ell(cB, c\beta) \geq 0$, for any $c\in \mathbb{R}$ we have:
\begin{equation*}
\begin{aligned}
F\big(\,c(B, \beta)&\,\big)\ \ \geq\ \ \lambda_N\|cB\|_* + \lambda_L\|W\circ cB\|_1 = \\
&|c| \big(\lambda_N\|B\|_* + \lambda_L\|W\circ B\|_1\big)\ \ \geq\ \ |c| \big(\,\lambda_N\|B\|_* + \lambda_L  W(i,j) |B(i,j)|\,\big). 
\end{aligned}
\end{equation*}
Since $\|B\|_*>0$ and $\lambda_L W(i,j) + \lambda_N > 0$, then also $\lambda_N\|B\|_* + \lambda_L W(i,j) |B(i,j)| > 0$. This implies
\begin{equation}
    \lim_{c\rightarrow +\infty} F(cB, c\beta) = \infty,
\end{equation}
hence $r$ is not a direction of recession. 
\subsection{Proof of Proposition~\ref{SymmetricSolution}}
\label{subsec:Prop22}
Suppose that $\big(B^*, \beta^*\big)$ is a solution to the minimization problem defined in equation~\eqref{LogisticSpinner}. Define the symmetrized matrix $\tilde{B} := \frac{1}{2}(B^* + B^{*\top})$. Then, for each observation $i = 1, \dots, n$, the linear predictor satisfies:
\begin{equation*}
\eta_i(\tilde{B},  \beta^*)\ \ =\ \  \Big\langle A_i, \frac{B^* + B^{*\top}}{2} \Big\rangle + X_i  \beta^*\ \ =\ \  \frac12 \langle A_i, B^* \rangle +  \frac12 \langle A_i, B^{*\top} \rangle + X_i  \beta^*\ \  =\ \  \eta_i(B^*,  \beta^*),
\end{equation*}
where we used the identity $\langle A_i, B^{*\top} \rangle = \langle A_i^\top, B^* \rangle = \langle A_i, B^* \rangle$ assuming that $A_i$ is symmetric. 

Consequently, $\ell(\tilde{B},  \beta^*) \ =\ \ell(B^*, \beta^*)$. The rest of the proof is identical as in~\cite{brzyski2024spinner} and uses the fact that the nuclear and $\ell_1$ norms are transpose-invariant, while $W$ is symmetric matrix.

\subsection{Proof of Proposition~\ref{PermutingSolution}}
\label{subsec:Prop23}
The proof of the statement is analogous to the one presented in~\cite{brzyski2024spinner}. 

We first observe that the matrix inner product $\langle A,\ B\rangle := \sum_{j,\ l} A_{j,l}B_{j,l}$ satisfies
\begin{equation}
    \langle A^{\pi},\ B \rangle = \sum_{j,l} (A^{\pi})_{j,l}B_{j,l} = \sum_{j,l}A_{\pi(j),\pi(l)}B_{j,l} = \sum_{j^{\prime},l^{\prime}}A_{j^{\prime},l^{\prime}}B_{\pi^{-1}(j^{\prime}),\pi^{-1}(l^{\prime})} = \langle A,\ B^{\pi^{-1}} \rangle,
\end{equation}
hence\quad $\eta_i(B, \beta\ |\ A_i^{\pi}) = \eta_i(B^{\pi^{-1}},\beta\ |\ A_i)\ = \eta_i(B^{\pi^{-1}},\beta)$.

The same reasoning holds for the $\ell_1$ norm of the Hadamard product: $\| W^{\pi} \circ B\|_1 = \| W \circ B^{\pi^{-1}} \|_1$. The nuclear norm of a matrix is invariant to any simultaneous permutation of rows and columns, since if $A = U\Lambda V^\top$ then $PAP^\top = P(U\Lambda V^\top)P^\top = PUP^\top(P\Lambda P^\top)P V^\top P^\top$, i.e., the set of singular values is unchanged.

Putting this all together, we obtain:
\begin{equation}
\label{permutationInvariance}
\begin{aligned}
F\big(&B,  \beta\ |\ A_i^{\pi}, W^{\pi}\big) =  \\
&- \sum_{i=1}^n \Big[y_i \eta_i(B^{\pi^{-1}}, \beta) - \ln\Big(1 + \exp\big(\eta_i(B^{\pi^{-1}}, \beta)\big)\Big)\Big]+\lambda_N \big\|B^{\pi^{-1}}\big\|_*+\lambda_L \big\| W \circ B^{\pi^{-1}}\big\|_1 \ =\ F\big(B^{\pi^{-1}}, \beta\ \big).
\end{aligned}
\end{equation}
Since \eqref{permutationInvariance} is true for all $(B,  \beta) \in \mathbb{R}^{p \times p} \times \mathbb{R}^d $, the optimal value of the permuted problem must equal that of the original problem:
\begin{equation}
\min_{B,\, \beta}\, F(B,  \beta\ |\ A_i^{\pi}, W^{\pi})\ =\ \min_{B,\, \beta}\, F(B, \beta)\ =\ F(\hat{B}, \hat{\beta}).
\end{equation}
Finally, applying equation~\eqref{permutationInvariance} to $\big(\widehat{B}^\pi, \widehat{\beta}\big)$, we get $F\big(\widehat{B}^\pi, \widehat{\beta}\ |\ A_i^{\pi}, W^{\pi}\big) = F\big(\hat{B}, \hat{\beta}\big)$, which confirms the optimality of $\big(\widehat{B}^\pi, \widehat{\beta}\big)$ and concludes the proof.











\end{document}